\documentclass[aps,twocolumn,superscriptaddress,amsmath]{revtex4}
\pdfoutput=1

\usepackage{dcolumn}
\usepackage{bm}
\usepackage{amsfonts}
\usepackage[pdftex]{color,graphicx}
\usepackage{amssymb}
\usepackage{rotate}
\usepackage{subfigure}
\usepackage{subfigure}
\usepackage{amsmath}
\usepackage{amsthm}
\usepackage{hyperref}

\newcommand{\ket}[1]{\ensuremath{\left| #1 \right\rangle}}
\newcommand{\bra}[1]{\ensuremath{\left\langle #1 \right|}}
\newcommand{\sand}[2]{\left\langle #1| #2\right\rangle}

\renewcommand{\-}{\,-\,}

\newcommand{\CC}{\mathbb{C}}

\newcommand{\HH}{\mathcal{H}}
\newcommand{\EE}{\mathbb{E}}

\newcommand{\up}{\uparrow}
\newcommand{\down}{\downarrow}

\newcommand{\set}[1]{\{ #1 \}}
\newcommand{\abs}[1]{\mathopen| #1 \mathclose|}	
\newcommand{\norm}[1]{\mathopen\| #1 \mathclose\|}
					
\newcommand{\Seq}[1]{\left\langle #1 \right\rangle}

\newcommand{\Paren}[1]{\left( #1 \right)}		

\let\oldmarginpar\marginpar
\renewcommand\marginpar[1]{\-\oldmarginpar[\raggedleft\tiny #1]%
{\raggedright\tiny #1}}

\newcommand{\agc}{\alpha_{\mathrm{gc}}}
\newcommand{\ahc}{\alpha_{\mathrm{hc}}}

\newcommand{\wpone}{\emph{w.p.}~1}

\newtheorem{thm}{Theorem}
\newtheorem{cor}[thm]{Corollary}

\newtheorem{lem}[thm]{Lemma}

\graphicspath{{figs/}}

\begin{document}


\title{Phase Transitions and Random Quantum Satisfiability}
\hypersetup{pdftitle={Phase Transitions and Random Quantum Satisfiability},
	pdfauthor={C. Laumann, R. Moessner, A. Scardicchio and S.L. Sondhi}}
	
\author{C. Laumann}
\affiliation{Department of Physics, Princeton University, Princeton, NJ 08544}

\author{R. Moessner}
\affiliation{Max Planck Institut f¨ur Physik Komplexer Systeme, %
	01187 Dresden, Germany}

\author{A. Scardicchio}
\affiliation{Department of Physics, Princeton University, Princeton, NJ 08544}
\affiliation{Princeton Center for Theoretical Science, %
	Princeton University, Princeton, NJ 08544}

\author{S. L. Sondhi}
\affiliation{Department of Physics, Princeton University, Princeton, NJ 08544}

\date{\today}


\begin{abstract}
  Alongside the effort underway to build quantum computers, it is important to better understand which classes of problems they will find easy and which others even they will find intractable. We study random ensembles of the QMA$_1$-complete quantum satisfiability (QSAT) problem introduced by Bravyi \cite{Bravyi:2006p4315}. QSAT appropriately generalizes the NP-complete classical satisfiability (SAT) problem. We show that, as the density of clauses/projectors is varied, the ensembles exhibit quantum phase transitions between phases that are satisfiable and unsatisfiable. Remarkably, almost all instances of QSAT for \emph{any} hypergraph exhibit the same dimension of the satisfying manifold. This establishes the QSAT decision problem as equivalent to a, potentially new, graph theoretic problem and that the hardest typical instances are likely to be localized in a bounded range of clause density.
\end{abstract}

\maketitle

\tableofcontents


\section{Introduction} 
\label{sec:intro}

The potential power of quantum computers motivates the intense effort
in progress to understand and, eventually, build them.  Much interest
has, naturally, been focused on algorithms that outperform their
classical counterparts. However the development of a complexity theory
for quantum computers suggests that we already know problems, those
shown to be QMA-complete, whose worst case solutions will take even
quantum computers a time exponential in problem size. These
appropriately generalize the class of NP-complete problems--those
which are easy to check but (believed to be) hard to solve on
classical computers--to quantum computers
\cite{Kitaev:1999ve,Kitaev:2002rm}.

The classic technique of complexity theory is assigning guilt by
association, i.e. by exhibiting the polynomial time equivalence of a
new problem to a particular problem believed not to be amenable to
efficient solution. For NP-complete problems this leads to the
celebrated Cook-Levin theorem which shows that the satisfiability
problem with clauses in three Boolean variables (3-SAT) encapsulates
the difficulty of the entire class \cite{Garey:1979ud}. While this is
a powerful and rigorous approach it has two limitations. It does not
directly tell us why a problem has hard instances and it does not tell
us what general features they possess.

Over the past decade joint work by computer scientists and statistical
physicists has produced an interesting set of insights into these two lacunae
for 3-SAT (and related constraint satisfaction problems). These insights have
come from the study of instances chosen at random from ensembles where the
density of clauses $\alpha$ acts as a control parameter. One can think of this
as representing the study of typical, but not necessarily worst, cases with a
specified density of clauses/constraints. Using techniques developed for the
study of random systems in physics, it has been shown that the ensembles
exhibit a set of {\it phase transitions} between a trivial satisfiable (SAT)
phase at small $\alpha$ and an unsatisfiable (UNSAT) phase at large $\alpha$
(see \emph{e.g.}
\cite{ScottKirkpatrick05271994,Monasson:1999p25,Mezard:2002p120,%
FlorentKrzakala06192007,Altarelli:2009ta}). These transitions mark sharp
discontinuities in the organization of SAT assignments in configuration space
as well a vanishing of SAT assignments altogether. This information has
provided a heuristic understanding of why known algorithms fail on random SAT
when the solution space is sufficiently complex and in doing so has localized
the most difficult members of these ensembles to a bounded range of $\alpha$.%

In this work we begin an analogous program for quantum computation
with the intention of gaining insight into the difficulty of solving
QMA-complete problems. Specifically, we introduce and study random
ensembles of instances of the quantum satisfiability ($k$-QSAT)
problem formulated by Bravyi \cite{Bravyi:2006p4315}.  Like classical
2-SAT, 2-QSAT is efficiently solvable (\emph{i.e.} it is in P), while
for $k \geq 4$, $k$-QSAT is QMA$_1$-complete~%
\footnote{QMA$_1$ is the one-sided version of QMA, in which ``yes''
  instances of problems always have proofs that are verifiable with
  probability 1, while ``no'' instances need only have bounded
  false-positive error rates. QMA allows the verifier to occasionally
  produce false-negatives as well.  This small difference is not
  believed to significantly influence the difficulty of the
  classes.}.
 
After laying out some important definitions and background on our problem in Sec.~\ref{sec:definition_of_random_qsat}, we attack the (``easy'') $k=2$ problem in Sec.~\ref{sec:2_qsat}. Here, we derive the complete phase diagram using transfer matrix techniques introduced by Bravyi \cite{Bravyi:2006p4315}. In the process, we discover that the classical geometry of the 2-QSAT interaction graph determines the generic satisfiability of random instances without reference to the random quantum Hamiltonian imposed upon the graph. We further show that our random ensemble asymptotically satisfies a technically important constraint known as the \emph{promise gap} with probability exponentially close to 1.

In Sec.~\ref{sec:k_qsat} we move on to establish rigorous bounds on the existence of SAT and UNSAT phases for the (``hard'') $k\ge 3$ problem. Although the proof is quite different from the analysis for $k=2$, we then directly recover the unexpected deduction that the satisfiability of a generic instance of $k$-QSAT reduces to a \emph{classical} graph property. This allows us to more tightly bound the UNSAT phase transition. We comment briefly on the straight-forward generalizations to some related ensembles and on the satisfaction of the promise gap. Finally, in Sec.~\ref{sec:conclusion} we lay out the goals for the continued study of random quantum problems and algorithms.


\section{Definition of random QSAT} 
\label{sec:definition_of_random_qsat}

We consider a set of $N$ qubits. We first randomly choose a collection
of $k$-tuples, $\{I_m,m=1\ldots M\}$ by independently including each
of the $\binom{N}{k}$ possible tuples with probability
$p=\alpha/\binom{N}{k}$. This defines an Erd\"os hypergraph with
$\alpha N$ expected edges; we exhibit simple examples of these for
$k=2,3$ in Fig.~\ref{fig:graphs}.  In classical $k$-SAT, the next step
is to generate an instance of the problem by further randomly
assigning a Boolean k-clause to each k-tuple of the hypergraph. In key
contrast to the classical case, where the Boolean variables and
clauses take on discrete values -- true or false -- their quantum
generalizations are continuous: the states of a qubit live in Hilbert
space, which allows for linear combinations of $|0\rangle$ (``false'')
and $|1\rangle$ (``true''). Thinking of a Boolean clause as forbidding
one out of $2^k$ configurations leads to its quantum generalization as
a projector $\Pi^I_\phi \equiv | \phi \rangle\langle \phi |$, which
penalizes any overlap of a state $|\psi\rangle$ of the $k$ qubits in
set $I$ with a state $|\phi\rangle$ in their $2^k$ dimensional Hilbert
space. In order to generate an instance of $k$-QSAT the states
$|\phi\rangle$ (of unit norm) will be picked randomly in this Hilbert
space.

The sum of these projectors defines a positive semidefinite
Hamiltonian $H = \sum\limits^{M}_{m = 1} \Pi^{I_m}_{\phi_m}$ and the
decision problem for a given instance is, essentially, to ask if there
exists a state that simultaneously satisfies all of the projectors,
\emph{i.e.}  to determine whether $H$ has ground state energy, $E_0$,
exactly zero. The qualifier ``essentially'' is needed because $E_0$ is
a continuous variable and in order that the problem be checkable by a
quantum verifier (and therefore in QMA$_1$), $H$ must be accompanied
by a \emph{promise} that $E_0$ is either exactly zero or exceeds an
$\epsilon \sim O(N^{-a})$ where $a$ is a constant. The scaling function $\epsilon$ is also known as the \emph{promise gap}. For our random
ensemble we wish to compute the statistics of this decision problem as
a function of $\alpha$. Specifically we would like to know if there
are phase transitions, starting with a SAT-UNSAT transition, in the
satisfying manifold as $\alpha$ is varied. Additionally, we would like
to check that the statistics in the large $N$ limit are dominated by
instances that automatically satisfy the promise gap.

We now note a few key properties of the hypergraph ensemble: At small
clause density $\alpha$, the size distribution of connected components
(``clusters'') is exponential. Moreover, almost all of these clusters
are treelike with a few, $O (N^0)$, containing a single closed loop
(see Fig.~\ref{fig:graphs}); the number of clusters with several
closed loops vanishes as $N \to \infty$. Above a critical value
$\agc(k) = \frac{1}{k(k-1)}$, a giant component emerges on which a
finite fraction of the qubits reside. Unlike the finite clusters, the
giant component may contain a non-vanishing (hyper)core, defined as
the subgraph remaining after recursively stripping away leaf nodes
(\emph{i.e.} nodes of degree 1). For $k=2$, an extensive hypercore
emerges continuously at $\ahc(2) = \agc(2) = \frac{1}{2}$; for
$k\ge3$, the hypercore appears abruptly with a finite fraction of the
nodes at $\alpha_{hc}(k) > \alpha_{gc}(k)$
\cite{Mezard:2003p5977,Molloy:2005lk}.

\begin{figure*}
  \includegraphics{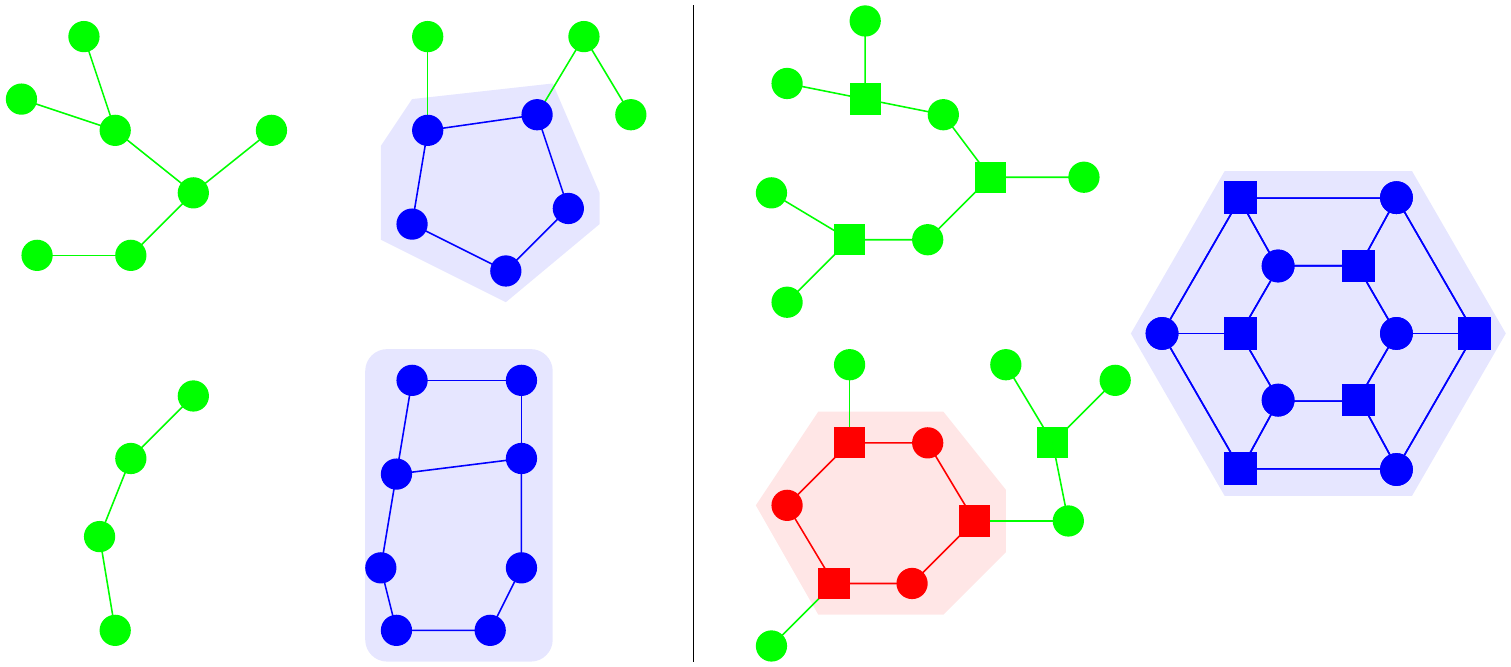}
  \caption{%
    Examples of random graphs and hypergraphs for (a) 2-SAT and (b)
    3-SAT, respectively. (a) The clusters, clockwise from bottom left,
    are chain, tree, clusters with one and two closed loops (``figure
    eight''). The short closed loops, as well as the planarity of the
    graphs, are not representative of the large-$N$ limit. The
    two-core of the graphs (shaded blue) is obtained by removing the
    (unshaded, green) dangling tree structures. (b) Each square
    represents a hyperedge connected to 3 nodes. Clockwise from top
    left are a tree, a hypercore and a hypergraph with simple loops
    (shaded red) but vanishing hypercore.}
  \label{fig:graphs}
\end{figure*}


\section{2-QSAT} 
\label{sec:2_qsat}

Let us first consider our questions in the specific context of 2-QSAT (phase
diagram in Fig.~\ref{fig:pd-2-qsat}). While this is a classically easy (P)
problem, the random ensemble has much of the structure we also find for the
harder $k>2$ case. An instance $H$ of QSAT is satisfiable if its kernel
$\ker(H)$ has nonzero dimension; in the following, we will find the kernel of $H$ explicitly by transforming the random projector problem almost surely into a simple Heisenberg ferromagnet, whose ground state space is well known.  

A central tool in this analysis is Bravyi's transfer matrix $T^{ij}_\phi$
which, given a vector $| \psi_i \rangle$ of the state of qubit $i$ yields
$|\psi_j \rangle = T^{ij}_\phi | \psi_i \rangle$ for qubit $j$, such that the
product state $|\psi_{ij} \rangle = |\psi_i \rangle \otimes |\psi_j \rangle$
satisfies the projector onto $\phi$: $\langle \psi_{ij} | \Pi^{ij}_\phi |
\psi_{ij} \rangle = 0$. Concretely, if the projector penalizes a joint state
of both qubits given by the complex vector $\phi = (\phi_{00}, \phi_{01},
\phi_{10}, \phi_{11})$, we have $T_\phi = \epsilon \phi^\dagger = \left(
+\phi^*_{01}~~+\phi^*_{11} \atop -\phi^*_{00}~~ -\phi^*_{10} \right)$. Here
$\epsilon$ is the standard antisymmetric matrix (Levi-Civita symbol) in two
dimensions. We note that the transfer matrix $T^{ij}$ for any given link is
almost surely invertible so this construction finds satisfying product states
for any input $\ket{\psi_i}$.

\subsection{Trees are SAT} 
\label{sub:trees_are_sat}

Consider the clusters that enter the 2-QSAT graph ensemble. Of these, a tree
comprising $n$ qubits and $n-1$ edges has a satisfying product state $|\Psi
\rangle = \bigotimes\limits^{n}_{j = 1} | \psi_j \rangle$, where $|\psi_j
\rangle$ is obtained from an arbitrary reference qubit $i=1$ by repeated
application of the $T$'s along the (unique) path joining $i$ with $j$. In
fact, the satisfying subspace is, almost always, $n + 1$ dimensional. To show
this, we will map the random projector problem directly onto the ferromagnetic
Hamiltonian on the same tree. In more mathematical terms, we construct a
non-unitary action of the permutation group on the qubit Hilbert space that
leaves the zero energy space $\ker(H)$ invariant. Thus, the ground state space
is precisely the totally symmetric space $\mathrm{Sym}^{n}\CC^2$, which has
dimension $n+1$. 

We define a particular, \emph{non-orthogonal} product basis for
the Hilbert space $\HH = \HH^0 \otimes \HH^1 \otimes\cdots\otimes
\HH^{n-1}$. First, choose any two linearly independent unit vectors
$\ket{\up^0}$ and $\ket{\down^0}$ as a basis for $\HH^0$. Second, use the
transfer matrix $T^{0j}$ to transfer this basis of $\HH^0$ to a normalized
basis 
\begin{eqnarray}
  \label{eq:7}
  \ket{\up^j} &=& \frac{T^{0j} \ket{\up^j}}{\norm{T^{1j} \ket{\up^j}}}
  \nonumber\\ 
  \ket{\down^j} &=& \frac{T^{0j} \ket{\down^j}}{\norm{T^{1j} \ket{\down^j}}}
\end{eqnarray}
of $\HH^j$ for each of the neighbors $j$ of $0$. Finally, recursively traverse
the tree to produce a pair of vectors $\ket{\up^i}, \ket{\down^i}$ for each
site $i$ in the tree. This procedures produces a choice of basis for each of
the individual qubit Hilbert spaces which we use to define a product basis we
call the \emph{transfer basis}. Although we use the notation of spin up and
down, we emphasize that the states need not be orthogonal.

Finally, we show that the ground state space of the Hamiltonian $H$ is
precisely the space of totally symmetric states defined with respect
to the transfer basis. We consider the constraint that a generic
vector $\Psi \in \HH$ is annihilated by $\Pi^{01}$ by factoring
$\Psi$:
\begin{eqnarray}
  \label{eq:psi-factor}
  \ket{\Psi} &=&  \ket{\up^0 \up^1}\ket{v_1^{2...n-1}} \nonumber \\
  &+& \ket{\down^0 \down^1}\ket{v_2^{2...n-1}}  \nonumber \\
  &+& (\ket{\up^0 \down^1} + \ket{\down^0 \up^1}) \ket{v_3^{2...n-1}} \nonumber \\
  &+& (\ket{\up^0 \down^1} -\ket{\down^0 \up^1}) \ket{v_4^{2...n-1}}.
\end{eqnarray}
The first three terms are identically annihilated by the projector
$\Pi^{01}$. This follows trivially from the construction of the
transfer basis for the first two terms while for the third term,
keeping track of indices carefully (summation over repeated indices is understood):
\begin{eqnarray}
  \label{eq:psi-symterm}
  &\Pi^{01}&\!\!\!\!(\ket{\up^0\down^1} + \ket{\down^0 \up^1}) = \ket{\phi^{01}} \big( (\phi^{01*})_{\alpha_0,\alpha_1}\up^0_{\alpha_0}\down^{1}_{\alpha_1} \nonumber\\
 && \quad+ (\phi^{01*})_{\alpha_0,\alpha_1}\down^0_{\alpha_0}\up^{1}_{\alpha_1}\big) \nonumber \\
  &=&\!\!\!\! \ket{\phi^{01}}
   (\down^{0T}\phi^{01*}\epsilon\phi^{01\dagger}\up^0 + \up^{0T}\phi^{01*}\epsilon\phi^{01\dagger}\down^0) \nonumber \\
  & = &\!\!\!\! 0 
\end{eqnarray}
where we have exploited the antisymmetry of $\epsilon$. Thus,
$\ket{v_4}$ must be zero and $\Psi$ clearly lies in the symmetric
eigenspace for swaps $(01)$. A similar argument holds for each of the
swaps $(i j)$ on links of the tree $B$. Since $B$ is connected, these
swaps generate the full permutation group $S_{n}$ and $\Psi$ must be
in the completely symmetric subspace of its action. In particular,
this means that the ground state space is isomorphic to
$\mathrm{Sym}^{n}\mathbb{C}^2$ which is $n+1$ dimensional.


\subsection{Loops} 
\label{sub:loops}

To understand hypergraphs with loops, we first look for satisfying product
states. For a cluster $G$ comprising $n_L$ independent closed loops
(\emph{i.e.} $n$ qubits with $n+n_L-1$ edges, see Fig.~\ref{fig:graphs}), a
product state is internally consistent only if the product of the $T$'s around
each closed loop returns the state it started with. For a graph with a single
closed loop, $O_1$, this imposes $\left(\prod_{jk\in O_1} T^{jk}\right)
|\psi_i \rangle = \lambda_1 | \psi_i \rangle$, which in general allows two
solutions, so that these graphs are SAT. By contrast, if a site $i$ is part of
a second independent closed loop, $O_2$, the additional demand
$\left(\prod_{jk\in O_2} T^{jk}\right) | \psi_i \rangle = \lambda_2 | \psi_i
\rangle$ already yields an overconstraint. Thus graphs with more than one
closed loop lack satisfying product states with probability 1.

In fact, we can generalize the above reasoning to show that it holds even for
the entangled (non-product) states in the problem as follows. We choose a
spanning tree on $G$ and starting qubit $0$ to define a transfer basis -- the
ground state space is necessarily a subspace of the associated symmetric
space. In particular, if $G$ contains a single closed loop $O_1$, we take
qubit $0$ on the loop. In this case, we choose the starting basis on qubit $0$
to satisfy the product state consistency condition for $O_1$
\begin{equation}
  \label{eq:12}
  \left(\prod_{m\in O_1} T^m\right) \psi^0 = \lambda \psi^0,
\end{equation}
\emph{i.e.} we take the eigenbasis of the loop transfer matrix, which
we denote $\up^0, \down^0$ with eigenvalues $\lambda_\up,
\lambda_\down$. These two states transfer to two linearly independent
product ground states for the loop -- all up and all down. A short
calculation verifies that in fact there are no more: write a generic
$\Psi \in \ker(H) \subset \mathrm{Sym}^{n}\mathbb{C}^2$ as
\begin{eqnarray}
  \label{eq:closed-only}
  \ket{\Psi} &=&  \ket{\up^{-1}\up^0}\ket{v_1} +
  \ket{\down^{-1}\down^0}\ket{v_2} \nonumber \\ 
  & + & \left( \ket{\up^{-1}\down^0} + \ket{\down^{-1}\up^0}\right)\ket{v_3},
\end{eqnarray}
where $-1$ is the qubit at the end of the loop $O_1$ and $v_i$ are
vectors on the $n-2$ other qubits. The projection $\Pi^{\Seq{-1,0}}$
annihilates the first two terms by choice of the basis. However: 
\begin{eqnarray}
  \label{eq:closed-anni}
  \Pi^{\Seq{-1,0}}\ket{\Psi} &=&  
  \ket{\phi^{\Seq{-1,0}}} (\up^{-1T}\phi^{{\Seq{-1,0}}\dagger}\down^0
  \nonumber\\ 
  &&\quad+ 
  \down^{-1T}\phi^{{\Seq{-1,0}}\dagger}\up^0) \ket{v_3} \nonumber \\
  &=&  \ket{\phi^{\Seq{-1,0}}} (\lambda_\up \up^{0T} \epsilon \down^0
  \nonumber\\ 
  &&\quad + \lambda_\down \down^{0T}\epsilon \up^0) \ket{v_3} \nonumber \\
  &=&  \ket{\phi^{\Seq{-1,0}}} (\lambda_\up - \lambda_\down)
  (\up^{0T} \epsilon \down^0) \ket{v_3}
\end{eqnarray}
Since $\lambda_\up \ne \lambda_\down$ \wpone{}, $v_3$ must be zero and
by symmetry, $\psi$ must be in the span of the all up and all down
states. 

An analogous calculation verifies that any additional closed loops
introduce projectors that are violated on this two-dimensional
subspace and thus graphs with more than one loop are unsatisfiable.

It is worthwhile to consider how the usual Heisenberg ferromagnet fits into
the above calculations. In this case, the transfer matrices are all the
identity and thus all of the loop constraint conditions are identically
satisfied. From the point of view of Eq.~\ref{eq:closed-anni}, the ferromagnet
has $\lambda_\up = \lambda_\down = 1$ and there is no reduction in the ground
state degeneracy due to the introduction of loop-closing ferromagnetic bonds.


\subsection{Phase diagram} 
\label{sub:phase_diagram}

\begin{figure}
  \centering
  \includegraphics[width=0.45\textwidth]{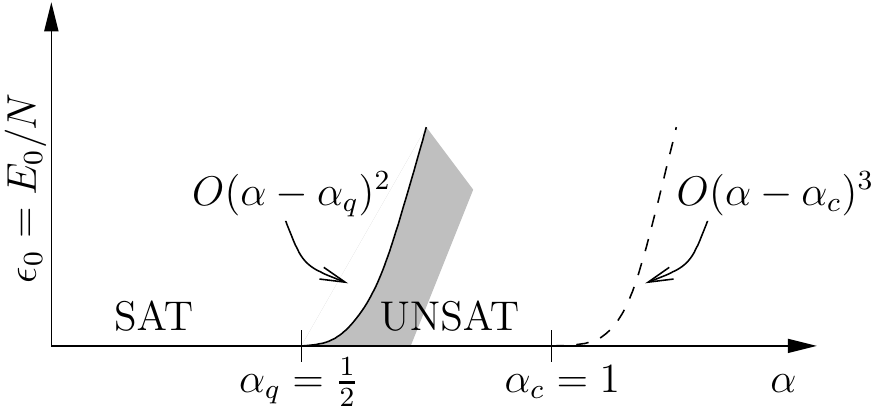}

  \caption{%
    Phase diagram of 2-QSAT. The quantum SAT-UNSAT transition
    $\alpha_q = \frac{1}{2}$ coincides with the emergence of a giant
    component in the random graph, which lies at half the classical
    2-SAT transition $\alpha_c = 1$ \cite{monasson1997smr}. The solid
    (dashed) line marks an asymptotic upper bound on the quantum
    (classical) ground state energy density $\epsilon_0$.}
  \label{fig:pd-2-qsat}
\end{figure}

In light of the above analysis, the existence of a SAT/UNSAT phase reduces to the presence of multiply connected components in the ensemble of 2-SAT graphs (rather than the combined ensemble of graphs and projectors). For $\alpha < \agc = \frac{1}{2}$, the number of such clusters vanishes in the limit of large $N$, so any instance is SAT with probability that goes to 1 as $N\to\infty$. At $\agc = 1/2$, closed loops proliferate as a giant component appears and thereafter all instances are UNSAT with probability 1. A straightforward upper bound on the energy in the UNSAT phase, $E \le O((\alpha - \agc)^2 N)$, follows from the fact that the fraction of sites in the core of the giant component grows as $(\alpha - \agc)^2$.

Physical intuition suggests that the ground state energy should be extensive above the transition and thus likely saturate the given upper bound. In the next section, we adduce strong evidence for this by exhibiting a minimal lower bound on the energy $E > N^{1-\epsilon}$ for any $\epsilon>0$ which holds except for exponentially rare instances. This follows from the twin claims that (i) the energy of a figure eight comprising $n$ sites decays only polynomially with $n$, and (ii) that the number of \emph{disjoint} figure eights in the random graph grows nearly linearly with $N$. Observe that this kind of lower bound is exactly what we need to establish that our ensemble keeps the promise that either $E = 0$ or $E > 1/N^a$ with probability exponentially close to 1. While this demonstration is not strictly needed for 2-QSAT (since it is in P), it is suggestive of what we might expect for the $k \ge 3$ cases.


\subsection{The promise gap for $k=2$: counting figure eights} 
\label{sec:count-figure-eights}

We expect on physical grounds that the UNSAT phase of $k$-QSAT has extensive
ground state energy with relatively vanishing fluctuations for any $k$. In
this case, the promise that $E\ge O(N^{-a})$ fails to be satisfied only with
exponentially small probability by Chebyshev's inequality. More
generally, so long as the average ground state energy is bounded below by a
polynomially small scale $E \ge O(N^{-b})$ with relatively vanishing
fluctuations, the promise will be violated with only exponentially small
probability for $a > b$.

Placing rigorous lower bounds on the expected quantum mechanical
ground state energy is generally difficult, but at least for $k=2$, we
argue as follows to find a nearly extensive bound scaling as
$N^{1-\epsilon}$ for any $\epsilon>0$. As the Hamiltonian for QSAT is a
sum of nonnegative terms, we can bound the ground state energy from
below by considering manageable subgraphs and ignoring the
contribution from other terms. In particular, the average ground
state energy of a figure eight graph, that is a loop of length $L$
with one additional crossing edge, is polynomially bounded below by
$O(L^{-\delta})$. %
This already gives a polynomial lower bound on the expected energy simply by
allowing $L$ to scale as $N$ and knowing that we will find at least one such
subgraph in the giant component with probability exponentially close to 1. We
will do better by finding a large number of disjoint such figure eight graphs.

The expected number of subgraphs $A$ in the random graph $G_{N,p}$ is given by
the following formula: 
\begin{equation}
  \label{eq:er-subgraph}
  \EE\#(A \subset G) = \frac{N!}{(N-\abs{A})! \abs{Aut(A)}}p^{e(A)}
\end{equation}
where $\abs{A}$ is the number of vertices in $A$, $e(A)$ is the number
of edges in $A$ and $Aut(A)$ is the group of automorphisms of $A$
($|Aut(A)|$ is the cardinality of this set). This formula simply
counts the number of ways of finding permutations of $\abs{A}$ nodes
in $G$ and connecting them up into an $A$ subgraph. For the fixed
clause density ensemble, we take $p = 2\alpha/N$.

A figure eight graph is uniquely specified by giving its size $L$ (we
take $L$ even) and the distance $d = 2 \dots L/2$ between the two
nodes that are connected by the crossing link. We let $A$ be the
disjoint union of $K$ figure eight graphs with $L$ nodes each and
cross bars at separation $L/2-1$. Such a graph has $|Aut(A)| = K! 2^K$
from the $K!$ permutations of the disjoint subgraphs and the two-fold
symmetry of each figure eight. Thus, the expected number of $K$-fold
disjoint figure eights of size $L$ is
\begin{equation}
  \label{eq:10}
  \EE\# = \frac{N!p^{K(L+1)}}{(N-K L)! K! 2^K}
\end{equation}
We now allow $K,L$ to scale with $N$ such that $KL \ll N$ and use
Stirling's formula to find the asymptotic entropy
\begin{equation}
  \label{eq:11}
  S \approx KL\left(\log 2\alpha - \frac{KL}{N}\right) + K
  \left(\log \alpha + 1-\log K - \log N\right).
\end{equation}
This entropy is positive and growing with $N$ so long as $L\gg\log N$,
$KL \ll N$ and $\alpha > \alpha_g = 1/2$. In this regime, assuming
that the fluctuations in $S$ are relatively small, we find that there
are $K\sim N^{1-\epsilon}$ disjoint figure eights of size $L\sim
N^{\epsilon/2}$ with exponentially high probability. These lead to a
nearly extensive lower bound on the expected ground state energy.



\section{$k$-QSAT} 
\label{sec:k_qsat}

We now turn to $k$-QSAT with $k >2$. As noted above, for sufficiently small
$\alpha$, the important hypergraphs have vanishing cores. In this regime, it
is possible to construct the hypergraph by adding in edges one by one where
each additional edge brings with it at least one new leaf node. Thus by a
generalization of the transfer matrix arguments for $k=2$, it is possible to
construct a satisfying product state on the full hypergraph; the details are
in Sec.~\ref{sub:existence-sat}. Deducing the existence of the UNSAT phase is
not as straightforward however for, unlike in the $k=2$ problem, we do not
know a set of unsatisfiable graphs which are present with finite probability
as $N \rightarrow \infty$; indeed these are not known for classical SAT
either. Thus, in Sec.~\ref{sub:the_unsat_phase_exists} we produce an indirect
proof that the dimension of the kernel of $H$ vanishes for sufficiently large
$\alpha$. 

Recall that for 2-QSAT the dimension of the satisfying manifold
ended up depending, with probability 1, \emph{only} on the topology of the
graph and not on the choices of the projectors. This is in fact a very general
result, which we prove directly in Sec.~\ref{sub:geometrization_theorem}. For \emph{any} hypergraph for $k$-QSAT (regardless of its likelihood)
the dimension of the satisfying manifold is independent of the choice of
projectors with probability 1. Moreover, non-generic choices of projectors
result in larger satisfying manifolds. Implicitly, this means that
satisfiability for the random ensemble is a graph theoretic property and
raises the extremely interesting possibility that random $k$-QSAT can be
formulated \emph{without reference to quantum Hamiltonians at all}.

Finally, we comment on the satisfaction of the promise gap in Sec.~\ref{sub:satisfying_the_promise} and on the generalization of QSAT to higher rank projectors in Sec.~\ref{sub:qsat_at_larger_rank}.

\subsection{Existence of SAT phase} 
\label{sub:existence-sat}

We prove the existence of a satisfiable phase at low clause density
for rank 1 $k$-QSAT by constructing product states of zero energy for
arbitrary hypergraphs containing a hypercore with a satisfying product
state. Since the core of a random graph vanishes for $\alpha
<\alpha_{\mathrm{hc}}(k)$, this proves the existence of a satisfiable phase
below the emergence of a hypercore.

\begin{lem}
  Suppose that $H$ is an instance of QSAT on $N$ qubits with a
  satisfying product state $\ket{\Psi}$. Let $H' = H + \Pi$ be an
  instance of QSAT on $N' > N$ qubits where $\Pi$ is a QSAT projector
  touching at least one qubit not among the original $N$. Then $H'$
  has a satisfying product state $\ket{\Psi'}$.
\end{lem}

\begin{proof}
  Assume first that the hyperedge $\Pi$ adjoins only one dangling
  qubit (untouched by $H$), say qubit 0. Let the components of
  $\ket{\Psi}$ be $\Psi_{i_1\cdots i_N}$. The key observation here is
  that the projector $\Pi$ with $k-1$ of its qubits fixed in a product
  state of $N$ qubits uniquely specifies the state of the dangling
  qubit $0$. If $\phi_{i_0 \cdots i_{k-1}}$ is the vector onto which $\Pi$
  projects, then the state of qubit 0 must be perpendicular to
  $\phi^*_{i_0 \cdots i_{k-1}} \Psi_{i_1 \cdots i_{k-1}}$.

  Thus, contraction with $\epsilon^{i'_0 i_0}$ defines a $k-1$ to $1$
  transfer matrix $T^{i'_0}_{i_1\cdots i_{k-1}}$ generalizing the
  Bravyi transfer matrix of 2-QSAT. The new product state given by
  $\ket{\Psi'} = \ket{\psi^0} \otimes \ket{\Psi}$ where $\psi^0_{i_0} =
  T^{i_0}_{i_1\cdots i_{k-1}} \Psi_{i_1 \cdots i_{k-1}}$ satisfies
  $\Pi$. Moreover, since $H$ does not act on qubit $0$, $H\ket{\Psi'}
  = \ket{\psi^0} \otimes H\ket{\Psi} = 0$. Thus $H' = H + \Pi$ is satisfied by
  the product state $\ket{\Psi'}$ on $N+1$ qubits.

  Finally, if $\Pi$ adjoins more than one dangling qubit, simply fix
  all but one of these to an arbitrary state and then apply the above
  transfer matrix procedure.
\end{proof}

\begin{thm}
  Suppose that $H$ is an instance of random QSAT and let $H_C$ be the
  restriction to its hypercore. If there is a satisfying product state
  on $H_C$, then there is almost surely a satisfying product state on
  $H$.
\end{thm}

\begin{proof}
  Let $\Pi_i$, $i= 1\cdots M(G) - M(C)$ be the sequence of projectors
  removed in the process of stripping leaf hyperedges from the
  hypergraph of $H$ to produce the hypercore. At each stage of this
  process, the removed projector $\Pi_i$ has at least one dangling
  qubit. Thus, if we iteratively reconstruct $H$ from $H_C$ by adding
  back the $\Pi_i$ in reverse order, we can apply the lemma at each
  stage to lift the product state on $H_C$ to a product state on $H$. 
\end{proof}


\subsection{The UNSAT phase exists} 
\label{sub:the_unsat_phase_exists}

For a given random hypergraph, let us construct $H$ via the sequence
$H_m = \sum\limits^{m}_{j = 1} \Pi^{I_j}_{\phi_j}$, i.e. we add the
projectors one at a time in some order. Clearly $H=H_{M}$. Let
$D_m$ be the dimension of the satisfying subspace for $H_m$; evidently
$D_0=2^N$. At each step, if the next added projector involves a set of
qubits that were not previously acted upon, then $D_{m+1}
= D_m (1- \frac{1}{2^k})$ as we may simply implement the projection by
reducing the size of a basis that can be factorized between the target
qubits and all others. It is intuitively plausible that this is
the best we can do---in cases where the target qubits are entangled
with the rest of the system we should expect to lose even more states,
i.e. 
$$
D_{m+1} \le D_m (1- \frac{1}{2^k})
$$
in general. A proof that this bound holds with probability 1 for projectors
randomly chosen according to the uniform Haar measure is contained in 
Appendix~\ref{sec:weak-unsat-bound}. From this we conclude that
$$
D_{\alpha N} \le 2^N (1- \frac{1}{2^k})^{\alpha N}.
$$
Thus, for $\alpha > -1/\log_2{(1- \frac{1}{2^k})} \sim 2^k$ the problem is
asymptotically almost always UNSAT.


\subsection{Geometrization theorem} 
\label{sub:geometrization_theorem}

\newtheorem*{geothm}{Geometrization Theorem}

\begin{geothm}
\label{thm:geometrization}
  Given an instance $H$ of random k-QSAT over a hypergraph $G$, the
  degeneracy of zero energy states $\dim(\ker(H))$ takes a particular
  value $D$ with probability 1 with respect to the choice of
  projectors on the edges of the hypergraph $G$.
\end{geothm}

\begin{proof}
  For a fixed hypergraph $G$ with $M$ edges, $H = H_\phi =
  \sum_{i=1}^{M} \Pi_i = \sum_{i=1}^{M} \ket{\phi_i}\bra{\phi_i}$ is a
  matrix valued function of the $2^k M$ components of the set of $M$
  vectors $\ket{\phi_i}$. In particular, its entries are polynomials
  in those components. Choose $\ket{\phi}$ such that $H$ has maximal
  rank $R$. Then there exists an $R\times R$ submatrix of $H$ such
  that $\det(H|_{R\times R})$ is nonzero. But this submatrix
  determinant is a polynomial in the components of $\ket{\phi}$ and
  therefore is only zero on a submanifold of the $\ket{\phi}$ of
  codimension at least 1. Hence, with probability 1, $H$ has rank $R$
  and the degeneracy $\dim(\ker(H)) = 2^N - R$.
\end{proof}

The theorem holds for general rank $r$ problems as well by a simple
modification of the argument to allow extra $\phi$'s to be associated
to each edge. 

A nice corollary of this result is a more stringent upper bound on the
size of the SAT phase. Consider any assignment of classical clauses on
a given hypergraph: we can also think of this as a special instance of
$k$-QSAT where the projectors are all diagonal in the computational
basis. As this is a non-generic choice of projectors, the dimension of
its satisfying manifold is an upper bound on the dimension for generic
choices. We conclude then that the classical UNSAT threshold is an
upper bound on the quantum threshold. Indeed, if we can identify the
most frustrated assignment of classical clauses on a given hypergraph,
i.e. the assignment that minimizes the number of satisfying
assignments, we could derive an even tighter bound.

\begin{cor}
\label{thm:class-bounds-quantum-nullity}
  The zero state degeneracy of the random quantum problem is bounded
  above (w.p.1) by the number of satisfying assignments of the most
  constrained classical $k$-SAT problem on the same graph.
\end{cor}

We have encountered many numerical examples in which the quantum problem has
fewer ground states than the most frustrated classical problem on the
same hypergraph. Thus, the bound of corollary
\ref{thm:class-bounds-quantum-nullity} is not tight. 

In the $k=2$ case, corollary \ref{thm:class-bounds-quantum-nullity}
gives us another way to show the critical point corresponds to the
emergence of a two-core: once a two-core exists, it will contain a
giant loop with two crossing bonds. This graph can be made
unsatisfiable by setting each of the clauses around the loop to
disallow the state $01$ and the crossing bonds to disallow $00$ and
$11$ respectively. Thus, the quantum problem is UNSAT \wpone. Ideally, we
could extend this geometric characterization to the $k>2$ problem, but
we have thus far failed to identify a family of unsatisfiable
hypergraphs that appear in the random hypergraph with non-vanishing
probability.  


\subsection{Satisfying the promise} 
\label{sub:satisfying_the_promise}

As in the case of 2-QSAT, physical experience suggests that the ground
state energy for $k$-QSAT should be extensive with small fluctuations
above the satisfiability transition. If this is true, the promise is
satisfied with probability exponentially close to 1, as in the $k=2$
case. Unfortunately, we do not know a set of unsatisfiable subgraphs
analogous to the figure eight's that might be used to provide a more
rigorous bound.

\begin{figure}
  \centering
  \includegraphics[width=0.45\textwidth]{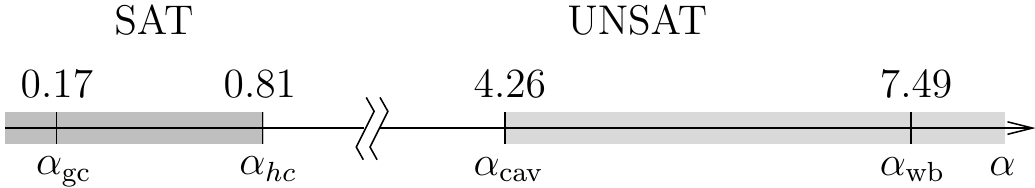}
  \caption{Phase diagram of 3-QSAT. For $\alpha<\ahc \approx 0.81$
    \cite{Molloy:2005lk}, the problem is rigorously SAT. The UNSAT
    transition is bounded below by the weak bound
    $\alpha_{\mathrm{wb}} = -1 / \log(1-1/2^k)$, but also (slightly
    less rigorously) by the classical cavity transition
    $\alpha_\mathrm{cav}$ \cite{Mezard:2002p120}. The giant component
    emerges squarely in the SAT phase at $\agc = \frac{1}{k(k-1)}
    \approx 0.17$.}
  \label{fig:pd-3-qsat}
\end{figure}


\subsection{QSAT at rank $r>2$} 
\label{sub:qsat_at_larger_rank}

We now briefly consider the extension of $k$-QSAT to $(k,r)$-QSAT, in
which the projectors $\Pi^I$ have rank $r$, i.e. they penalize a
uniformly chosen $r$-dimensional subspace of the $2^k$ dimensional $k$
qubit Hilbert space. The main results regarding $k$-QSAT generalize
naturally: satisfiability is almost surely only dependent on the
underlying graph and the weak bound on the ground state degeneracy
becomes $D^{(r)}_{\alpha N} \leq 2^N \left( 1 - \frac{r}{2^k}
\right)^{\alpha N}$, implying a bound $\alpha^{(r)}_c \leq -1/\log_2
\left( 1 - \frac{r}{2^k} \right)$. However, there need not be a SAT
phase at all: if $r=2^k$, the projectors $\Pi^I$ are each the identity
and the ground state energy is $\alpha N > 0$ for any positive
$\alpha$. More generally, there is some critical rank $r_c$ above
which the SAT phase disappears. 

We bound $r_c$ above by $2^k/2$ by exhibiting an unsatisfiable
subgraph of the random hypergraph that arises asymptotically almost surely ($N\to\infty$) in the random graph ensemble even at small
$\alpha$. Consider a chain with two clauses and one shared
qubit. Classically, this corresponds to a two-clause problem on $2k-1$
bits where we allow each clause to forbid $r=2^k/2$
configurations. Now let the first clause forbid all configurations in
which the shared bit is $0$ and the second clause all configurations
in which it is $1$. This classical problem is clearly unsatisfiable
and therefore so is any quantum problem on this subgraph
\wpone{}. Indeed, since there are extensively many such small chain
components in the hypergraph, each with an independently chosen
$O(N^0)$ ground state energy, this provides an extensive lower bound
on the total ground state energy.

The bound $r_c \le 2^k/2$ is not tight however. For example, for
$k=3$, an open chain of length four with rank $3< 2^3/2 = 4$ is
quantum mechanically unsatisfiable \wpone{}, as can be checked
numerically, and therefore so is the ensemble for $(k,r) =
(3,3)$. However, there is no classically unsatisfiable problem on this
chain and it is harder to construct a rigorous bound that scales with
$k$ using this starting point.



\section{Conclusions and Outlook} 
\label{sec:conclusion}

As in the classical work that was our inspiration, we expect that the
phase structure that we have identified also shows up in the actual
performance of potential quantum algorithms. Specifically, the
classical results and the general intuition from the study of quantum
phase transitions \cite{Sondhi:1997pj,Sachdev:2001ms} both strongly
suggest that the decision problem will be hardest near the SAT-UNSAT
boundary and easy away from such a bounded range. This expectation is
given further force by the graph theoretic formulation of typical
instances whereby we may again expect that graphs with an intermediate
clause density are the ones where it might prove hardest to decipher
the relevant property. As a first step in testing this expectation, we
have begun an investigation of the simplest adiabatic algorithm for
this problem \cite{Farhi:2001p802} and the relevant excitation
spectra.

Looking ahead, the immediate challenge is to explicitly identify the
graph theoretic property that characterizes satisfiable instances of
random QSAT without reference to the quantum mechanical problem. If
this can be accomplished it will shed light on the difficulty of
deciding generic instances of random QSAT and provide a surprising
connection to classical complexity classes. Separately, we intend to
investigate whether there are analogs of the clustering transitions of
classical $k$-SAT in the quantum problem and whether generalizations
of the cavity method to quantum problems
\cite{Laumann:2008zl,Hastings:2007p5560,Leifer:2008pg1899} can help
locate them. Finally, we would like to translate this improved
understanding of this ensemble of problems into new algorithms for
solving them on the lines of the belief/survey propagation algorithms
for classical SAT.


\section*{Acknowledgements} 

We thank S. Bravyi, B. Terhal, E. Lieb, U. Vazirani and S. Hughes for
discussions and feedback on this work; and, the MITRE/PCTS quantum computation
program for support.

\bibliography{qsat}

\begin{thebibliography}{18}
\expandafter\ifx\csname natexlab\endcsname\relax\def\natexlab#1{#1}\fi
\expandafter\ifx\csname bibnamefont\endcsname\relax
  \def\bibnamefont#1{#1}\fi
\expandafter\ifx\csname bibfnamefont\endcsname\relax
  \def\bibfnamefont#1{#1}\fi
\expandafter\ifx\csname citenamefont\endcsname\relax
  \def\citenamefont#1{#1}\fi
\expandafter\ifx\csname url\endcsname\relax
  \def\url#1{\texttt{#1}}\fi
\expandafter\ifx\csname urlprefix\endcsname\relax\def\urlprefix{URL }\fi
\providecommand{\bibinfo}[2]{#2}
\providecommand{\eprint}[2][]{\url{#2}}

\bibitem[{\citenamefont{Bravyi}(2006)}]{Bravyi:2006p4315}
\bibinfo{author}{\bibfnamefont{S.}~\bibnamefont{Bravyi}},
  \bibinfo{journal}{arXiv}  (\bibinfo{year}{2006}), \eprint{quant-ph/0602108}.

\bibitem[{\citenamefont{Kitaev}(1999)}]{Kitaev:1999ve}
\bibinfo{author}{\bibfnamefont{A.~Y.} \bibnamefont{Kitaev}}, in
  \emph{\bibinfo{booktitle}{AQIP'99}} (\bibinfo{organization}{DePaul
  University}, \bibinfo{year}{1999}).

\bibitem[{\citenamefont{Kitaev et~al.}(2002)\citenamefont{Kitaev, Shen, and
  Vyalyi}}]{Kitaev:2002rm}
\bibinfo{author}{\bibfnamefont{A.~Y.} \bibnamefont{Kitaev}},
  \bibinfo{author}{\bibfnamefont{A.}~\bibnamefont{Shen}}, \bibnamefont{and}
  \bibinfo{author}{\bibfnamefont{M.~N.} \bibnamefont{Vyalyi}},
  \emph{\bibinfo{title}{Classical and quantum computation}},
  vol.~\bibinfo{volume}{47} of \emph{\bibinfo{series}{Graduate studies in
  mathematics}} (\bibinfo{publisher}{American Mathematical Society},
  \bibinfo{address}{Providence, R.I.}, \bibinfo{year}{2002}).

\bibitem[{\citenamefont{Garey and Johnson}(1979)}]{Garey:1979ud}
\bibinfo{author}{\bibfnamefont{M.~R.} \bibnamefont{Garey}} \bibnamefont{and}
  \bibinfo{author}{\bibfnamefont{D.~S.} \bibnamefont{Johnson}},
  \emph{\bibinfo{title}{Computers and Intractability : A Guide to the Theory of
  NP-Completeness}}, Series of Books in the Mathematical Sciences
  (\bibinfo{publisher}{W. H. Freeman \& Co Ltd}, \bibinfo{year}{1979}).

\bibitem[{\citenamefont{Kirkpatrick and
  Selman}(1994)}]{ScottKirkpatrick05271994}
\bibinfo{author}{\bibfnamefont{S.}~\bibnamefont{Kirkpatrick}} \bibnamefont{and}
  \bibinfo{author}{\bibfnamefont{B.}~\bibnamefont{Selman}},
  \bibinfo{journal}{Science} \textbf{\bibinfo{volume}{264}},
  \bibinfo{pages}{1297} (\bibinfo{year}{1994}).

\bibitem[{\citenamefont{Monasson et~al.}(1999)\citenamefont{Monasson, Zecchina,
  Kirkpatrick, Selman, and Troyansky}}]{Monasson:1999p25}
\bibinfo{author}{\bibfnamefont{R.}~\bibnamefont{Monasson}},
  \bibinfo{author}{\bibfnamefont{R.}~\bibnamefont{Zecchina}},
  \bibinfo{author}{\bibfnamefont{S.}~\bibnamefont{Kirkpatrick}},
  \bibinfo{author}{\bibfnamefont{B.}~\bibnamefont{Selman}}, \bibnamefont{and}
  \bibinfo{author}{\bibfnamefont{L.}~\bibnamefont{Troyansky}},
  \bibinfo{journal}{Nature} \textbf{\bibinfo{volume}{400}},
  \bibinfo{pages}{133} (\bibinfo{year}{1999}).

\bibitem[{\citenamefont{Mezard et~al.}(2002)\citenamefont{Mezard, Parisi, and
  Zecchina}}]{Mezard:2002p120}
\bibinfo{author}{\bibfnamefont{M.}~\bibnamefont{Mezard}},
  \bibinfo{author}{\bibfnamefont{G.}~\bibnamefont{Parisi}}, \bibnamefont{and}
  \bibinfo{author}{\bibfnamefont{R.}~\bibnamefont{Zecchina}},
  \bibinfo{journal}{Science} \textbf{\bibinfo{volume}{297}},
  \bibinfo{pages}{812} (\bibinfo{year}{2002}).

\bibitem[{\citenamefont{Krzakala et~al.}(2007)\citenamefont{Krzakala,
  Montanari, Ricci-Tersenghi, Semerjian, and
  Zdeborov{\'a}}}]{FlorentKrzakala06192007}
\bibinfo{author}{\bibfnamefont{F.}~\bibnamefont{Krzakala}},
  \bibinfo{author}{\bibfnamefont{A.}~\bibnamefont{Montanari}},
  \bibinfo{author}{\bibfnamefont{F.}~\bibnamefont{Ricci-Tersenghi}},
  \bibinfo{author}{\bibfnamefont{G.}~\bibnamefont{Semerjian}},
  \bibnamefont{and}
  \bibinfo{author}{\bibfnamefont{L.}~\bibnamefont{Zdeborov{\'a}}},
  \bibinfo{journal}{Proceedings of the National Academy of Sciences}
  \textbf{\bibinfo{volume}{104}}, \bibinfo{pages}{10318}
  (\bibinfo{year}{2007}).

\bibitem[{\citenamefont{Altarelli et~al.}(2009)\citenamefont{Altarelli,
  Monasson, Semerjian, and Zamponi}}]{Altarelli:2009ta}
\bibinfo{author}{\bibfnamefont{F.}~\bibnamefont{Altarelli}},
  \bibinfo{author}{\bibfnamefont{R.}~\bibnamefont{Monasson}},
  \bibinfo{author}{\bibfnamefont{G.}~\bibnamefont{Semerjian}},
  \bibnamefont{and} \bibinfo{author}{\bibfnamefont{F.}~\bibnamefont{Zamponi}},
  in \emph{\bibinfo{booktitle}{Handbook of Satisfiability}}
  (\bibinfo{publisher}{IOS press}, \bibinfo{year}{2009}), vol.
  \bibinfo{volume}{185} of \emph{\bibinfo{series}{Frontiers in Artificial
  Intelligence and Applications}}, \eprint{arXiv:0802.1829}.

\bibitem[{\citenamefont{M{\'e}zard et~al.}(2003)\citenamefont{M{\'e}zard,
  Ricci-Tersenghi, and Zecchina}}]{Mezard:2003p5977}
\bibinfo{author}{\bibfnamefont{M.}~\bibnamefont{M{\'e}zard}},
  \bibinfo{author}{\bibfnamefont{F.}~\bibnamefont{Ricci-Tersenghi}},
  \bibnamefont{and} \bibinfo{author}{\bibfnamefont{R.}~\bibnamefont{Zecchina}},
  \bibinfo{journal}{J. Stat. Phys.} \textbf{\bibinfo{volume}{111}},
  \bibinfo{pages}{505} (\bibinfo{year}{2003}).

\bibitem[{\citenamefont{Molloy}(2005)}]{Molloy:2005lk}
\bibinfo{author}{\bibfnamefont{M.}~\bibnamefont{Molloy}},
  \bibinfo{journal}{Random Struct. Algorithms} \textbf{\bibinfo{volume}{27}},
  \bibinfo{pages}{124} (\bibinfo{year}{2005}).

\bibitem[{\citenamefont{Monasson and Zecchina}(1997)}]{monasson1997smr}
\bibinfo{author}{\bibfnamefont{R.}~\bibnamefont{Monasson}} \bibnamefont{and}
  \bibinfo{author}{\bibfnamefont{R.}~\bibnamefont{Zecchina}},
  \bibinfo{journal}{Phys. Rev. E} \textbf{\bibinfo{volume}{56}},
  \bibinfo{pages}{1357} (\bibinfo{year}{1997}).

\bibitem[{\citenamefont{Sondhi et~al.}(1997)\citenamefont{Sondhi, Girvin,
  Carini, and Shahar}}]{Sondhi:1997pj}
\bibinfo{author}{\bibfnamefont{S.~L.} \bibnamefont{Sondhi}},
  \bibinfo{author}{\bibfnamefont{S.~M.} \bibnamefont{Girvin}},
  \bibinfo{author}{\bibfnamefont{J.~P.} \bibnamefont{Carini}},
  \bibnamefont{and} \bibinfo{author}{\bibfnamefont{D.}~\bibnamefont{Shahar}},
  \bibinfo{journal}{Rev. Mod. Phys.} \textbf{\bibinfo{volume}{69}},
  \bibinfo{pages}{315} (\bibinfo{year}{1997}).

\bibitem[{\citenamefont{Sachdev}(2001)}]{Sachdev:2001ms}
\bibinfo{author}{\bibfnamefont{S.}~\bibnamefont{Sachdev}},
  \emph{\bibinfo{title}{Quantum Phase Transitions}}
  (\bibinfo{publisher}{Cambridge University Press},
  \bibinfo{address}{Cambridge, UK}, \bibinfo{year}{2001}).

\bibitem[{\citenamefont{Farhi et~al.}(2001)\citenamefont{Farhi, Goldstone,
  Gutmann, Lapan, Lundgren, and Preda}}]{Farhi:2001p802}
\bibinfo{author}{\bibfnamefont{E.}~\bibnamefont{Farhi}},
  \bibinfo{author}{\bibfnamefont{J.}~\bibnamefont{Goldstone}},
  \bibinfo{author}{\bibfnamefont{S.}~\bibnamefont{Gutmann}},
  \bibinfo{author}{\bibfnamefont{J.}~\bibnamefont{Lapan}},
  \bibinfo{author}{\bibfnamefont{A.}~\bibnamefont{Lundgren}}, \bibnamefont{and}
  \bibinfo{author}{\bibfnamefont{D.}~\bibnamefont{Preda}},
  \bibinfo{journal}{Science} \textbf{\bibinfo{volume}{292}},
  \bibinfo{pages}{472} (\bibinfo{year}{2001}).

\bibitem[{\citenamefont{Laumann et~al.}(2008)\citenamefont{Laumann,
  Scardicchio, and Sondhi}}]{Laumann:2008zl}
\bibinfo{author}{\bibfnamefont{C.}~\bibnamefont{Laumann}},
  \bibinfo{author}{\bibfnamefont{A.}~\bibnamefont{Scardicchio}},
  \bibnamefont{and} \bibinfo{author}{\bibfnamefont{S.~L.}
  \bibnamefont{Sondhi}}, \bibinfo{journal}{Phys. Rev. B}
  \textbf{\bibinfo{volume}{78}}, \bibinfo{pages}{134424}
  (\bibinfo{year}{2008}).

\bibitem[{\citenamefont{Hastings}(2007)}]{Hastings:2007p5560}
\bibinfo{author}{\bibfnamefont{M.~B.} \bibnamefont{Hastings}},
  \bibinfo{journal}{Phys. Rev. B} \textbf{\bibinfo{volume}{76}}
  (\bibinfo{year}{2007}).

\bibitem[{\citenamefont{Leifer and Poulin}(2008)}]{Leifer:2008pg1899}
\bibinfo{author}{\bibfnamefont{M.}~\bibnamefont{Leifer}} \bibnamefont{and}
  \bibinfo{author}{\bibfnamefont{D.}~\bibnamefont{Poulin}},
  \bibinfo{journal}{Ann. Phys.} \textbf{\bibinfo{volume}{323}},
  \bibinfo{pages}{1899} (\bibinfo{year}{2008}).

\end{thebibliography}

\appendix

\section{Weak Unsat Bound} 
\label{sec:weak-unsat-bound}

\newtheorem*{unsatboundthm}{Weak UNSAT Bound}

\begin{unsatboundthm}
  Given an instance of random $(k,r)$-QSAT with Hamiltonian $H$ over a hypergraph $G$ with $N$ qubits and $M$ clauses, the degeneracy of zero energy states (i.e.\ the number of satisfying assigments)
  $D=\dim(\ker(H))$ is bounded above by
  \begin{equation}
    \label{eq:1}
    D \le 2^N\Paren{1-\frac{r}{2^k}}^M
  \end{equation}
  with probability 1.
\end{unsatboundthm}

\begin{proof}
  For a given random hypergraph $G$, we construct $H = H_M$ via
  the sequence $H_m = \sum_{j=1}^{m} \Pi_{j}$, where $m = 1
  \dots M$. That is, we add projectors one at a time in some fixed
  order. Let $D_m$ be the degeneracy of zero energy states after $m$
  projectors have been added. Clearly,
  \begin{equation}
    \label{eq:2}
    D_0 = 2^N.
  \end{equation}
  The result will follow by induction after we show that each
  additional projector reduces the degeneracy by at least a factor
  $1-\frac{r}{2^k}$ \wpone. That is, 
  \begin{equation}
    \label{eq:3}
    D_{m+1} \le D_m \Paren{1 - \frac{r}{2^k}}.
  \end{equation}

  We consider adding a projector $\Pi=\sum_{\alpha=1}^r
  \ket{\phi^\alpha}\bra{\phi^\alpha}$, where $\sand{\phi^\alpha}{\phi^\beta}=\delta_{\alpha\beta}$, to the Hamiltonian $H$ to construct
  $H'$. By appropriate reordering, we can always assume that $\Pi$
  acts on the first $k$ qubits of the full $N$ qubit Hilbert space. We
  choose a basis for the zero energy subspace $\ker(H)$ prior to the
  addition of $\Pi$ as follows
  \begin{equation}
    \ket{d} = \sum_{i=1}^L \ket{a^i_d} \ket{l^i}
  \end{equation}
  where $d \in \set{1,...,D}$ labels the basis states, 
  $\ket{a^i_d}\in\mathbb{C}^{2^k}$ are vectors in the first $k$ qubit
  factor of the Hilbert space and $\ket{l^i}$ are $L$ linearly
  independent vectors on the remaining $(N-k)$ qubit factor. 

  The kernel of $H' = H + \Pi$ is the subspace of the kernel of
  $H$ which is annihilated by $\Pi$. We write a generic
  vector of $\ker(H)$ as:
  \begin{equation}
    \label{eq:4}
    \ket{\psi} = \sum_{d=1}^D \psi_d \ket{d}.
  \end{equation}
  Annihilation by $\Pi$ leads to the condition
  \begin{equation}
    \label{eq:5}
    0 = \Pi\ket{\psi} = 
    \sum_{i=1}^L\sum_{\alpha=1}^r \ket{\phi^\alpha} \ket{l^i} \sum_d \sand{\phi^\alpha}{a^i_d} \psi_d,
  \end{equation}
  which by linear independence of the $\ket{\phi^\alpha}\ket{l^i}$ requires, for any $\alpha, i$:
  \begin{equation}
    \label{eq:6}
    0 = \sum_{d=1}^D \sand{\phi^\alpha}{a^i_d} \psi_d.
  \end{equation}
  That is, $\psi_d$ must lie in the kernel of the $r L\times D$ matrix
  $A_{(i\alpha)d}=\sand{\phi^\alpha}{a^i_d}$. We now claim that $A_{(i\alpha)d}$
  has rank $R$ at least $r D/2^k$ with probability 1, which will prove
  our inductive step. This follows from two observations:

  \begin{enumerate}
  \item The rank $R$ of $A$ is a bounded random variable that takes on
    its maximal value over the choice of $\ket{\phi^\alpha}$ with
    probability 1. $A$ can be viewed as a matrix of monomials in the
    components of $\ket{\phi^\alpha}$. Choose an orthonormal frame
    $\ket{\phi^\alpha}$ maximizing the rank $R$ of $A$; with this
    choice, there exists some $R\times R$ submatrix of $A$ such that
    $\det(A|_{R\times R})$ is nonzero. But this submatrix determinant
    is a homogenous polynomial in the components of
    $\ket{\phi^\alpha}$ and therefore is only zero on a submanifold of
    codimension at least 1 in the complex Grassmannian
    $\mathrm{Gr}(r,2^k)$ (\emph{i.e.} the space of choices of rank $r$
    projectors $\Pi$).  
    Hence, almost every matrix $A$ will have maximal rank $R$.

  \item We now need only exhibit one set of $\ket{\phi^\alpha}$'s such that
    $A$ has rank $R \ge r D/2^k$. This is easy: at least one $R$-tuple
    of the standard basis vectors will provide such an $A$. Consider
    the matrix of vectors $\ket{a^i_d}$ and let
    $r_1,r_2,\cdots,r_{2^k}$ be the ranks of each of the $2^k$
    component matrices (\emph{ie.}  the matrices $A$ obtained by using
    the standard basis elements for $\ket{\phi}$). The $r_j$ are the
    number of linearly independent rows in each of these matrices.
    Now concatenate each of these component matrices vertically into a
    giant $2^k L \times D$ matrix. The row rank of this matrix can be
    no larger than $r_1 + r_2 + \cdots + r_{2^k}$ by construction, but
    there must be a full $D$ linearly independent columns: any linear
    relation on the columns of the giant matrix corresponds to a
    linear relation $\sum_d w_d \ket{a^i_d} = 0$ which lifts trivially
    to a linear relation among the basis vectors $\ket{d}$. Hence,
    \begin{equation*}
      r_1+r_2+\cdots+r_{2^k} \ge D.
    \end{equation*}
    From this relation, there must exist some collection $r_{i_j}$ of
    size $r$ such that $r_{i_1}+\cdots+r_{i_r}\ge r D / 2^k$. This
    collection of basis vectors provides the frame we
    desire.

  \end{enumerate}
\end{proof}


\end{document}